\documentclass[a4paper,USenglish,cleveref,autoref,numberwithinsect]{lipics-v2021} 

\usepackage{microtype}
\usepackage{complexity}

\hideLIPIcs
\nolinenumbers
\graphicspath{{figs/}}
\title{Plane Hamiltonian Cycles in Convex Drawings} 

\author{Helena Bergold}{Institut f{\"u}r Informatik, Freie Universit{\"a}t Berlin, Germany}{helena.bergold@fu-berlin.de}{https://orcid.org/0000-0002-9622-8936}{Supported by DFG Grant DFG-GRK~2434.}

\author{Stefan Felsner}{Institut f{\"u}r Mathematik, Technische Universit{\"a}t Berlin, Germany}{felsner@math.tu-berlin.de}{https://orcid.org/0000-0002-6150-1998}{Partially supported by DFG Grant FE~340/13-1.}

\author{Meghana M. Reddy}{Department of Computer Science, ETH Z{\"u}rich, Switzerland}{meghana.mreddy@inf.ethz.ch}{https://orcid.org/0000-0001-9185-1246}{Supported by the Swiss National Science Foundation within the collaborative DACH project Arrangements and Drawings as SNSF Project 200021E-171681.}

\author{Joachim Orthaber}{Institute of Algorithms and Theory, Graz University of Technology, Austria}{orthaber@tugraz.at}{https://orcid.org/0000-0002-9982-0070}{Supported by the Austrian Science Fund (FWF) grant W1230.}

\author{Manfred Scheucher}{Institut f{\"u}r Mathematik, Technische Universit{\"a}t Berlin, Germany}{scheucher@math.tu-berlin.de}{https://orcid.org/0000-0002-1657-9796}{Supported by DFG Grant SCHE~2214/1-1.}

\authorrunning{H.~Bergold, S.~Felsner, M.~M.~Reddy, J.~Orthaber, and M.~Scheucher}

\ccsdesc[500]{Mathematics of computing~Combinatorics}
\ccsdesc[500]{Mathematics of computing~Graph theory}
\ccsdesc[500]{Theory of computation~Computational geometry}

\keywords{simple drawing, convexity hierarchy, plane pancyclicity, plane Hamiltonian connectivity, maximal plane subdrawing}

\relatedversion{Earlier versions of this article appeared at SoCG'24~\cite{BFROS_SOCG24} and in Bergold's thesis~\cite{Bergold_Diss23}.}

\usepackage{hyperref}
\usepackage[font=small,labelfont=bf]{caption}
\usepackage[font=small,labelfont=normalfont,labelformat=simple]{subcaption}

\newcommand{\calD}{\mathcal{D}}
\newcommand{\calC}{\mathcal{C}}
\newcommand{\bigO}{\mathcal{O}}

\crefname{conjecture}{conjecture}{conjectures}

\begin{document}

\maketitle

\begin{abstract}
A conjecture by Rafla from 1988 asserts that every simple drawing of the complete graph $K_n$ admits a plane Hamiltonian cycle.
It turned out that already the existence of much simpler non-crossing substructures in such drawings is hard to prove.
Recent progress was made by Aichholzer et al.\ and by Suk and Zeng who proved the existence of a plane path of length $\Omega(\log n / \log \log n)$ and of a plane matching of size $\Omega(n^{1/2})$ in every simple drawing of $K_n$.

Instead of studying simpler substructures, we prove Rafla's conjecture for the subclass of convex drawings, the most general class in the convexity hierarchy introduced by Arroyo et al.
Moreover, we show that every convex drawing of $K_n$ contains a plane Hamiltonian path between each pair of vertices (Hamiltonian connectivity) and a plane $k$-cycle for each $3 \leq k \leq n$ (pancyclicity), and present further results on maximal plane subdrawings.
\end{abstract}

\section{Introduction}
\label{sec:intro}

Starting from Tur{\'a}n's brick factory problem from the 1940's, which initiated the study of crossing-minimal drawings, simple drawings gained a lot of attention.
In a \emph{simple drawing}\footnote{%
In the literature, simple drawings are also called \emph{good drawings} or \emph{simple topological graphs}.}
of a graph in the plane (resp.\ on the sphere), the vertices are mapped to distinct points, and edges are drawn as simple curves such that they connect the two corresponding end-vertices but do not contain other vertices.
Moreover, two edges share at most one common point, which is either a common vertex or a proper crossing.
This in particular excludes touchings.
\Cref{fig:simple_obstructions} shows these forbidden patterns.

\begin{figure}[htb]
\centering
\subcaptionbox{\label{fig:simple_obstructions_selfint}}[.18\textwidth]{\includegraphics[page=1]{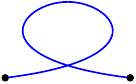}}
\subcaptionbox{\label{fig:simple_obstructions_through-vertex}}[.18\textwidth]{\includegraphics[page=2]{simple_obstructions_dcg.pdf}}
\subcaptionbox{\label{fig:simple_obstructions_touching}}[.18\textwidth]{\includegraphics[page=3]{simple_obstructions_dcg.pdf}}
\subcaptionbox{\label{fig:simple_obstructions_doublecross}}[.18\textwidth]{\includegraphics[page=4]{simple_obstructions_dcg.pdf}}
\subcaptionbox{\label{fig:simple_obstructions_adjcross}}[.18\textwidth]{\includegraphics[page=5]{simple_obstructions_dcg.pdf}}
\caption{In a simple drawing, edges are not allowed to \subref{fig:simple_obstructions_selfint}~cross themselves or \subref{fig:simple_obstructions_through-vertex}~pass through vertices.
 \subref{fig:simple_obstructions_touching}~If two edges meet in their relative interior, they have to cross (no touchings).
 \subref{fig:simple_obstructions_doublecross}~Each pair of edges crosses at most once and \subref{fig:simple_obstructions_adjcross}~adjacent edges do not cross.}
\label{fig:simple_obstructions}
\end{figure}

Instead of minimizing the number of crossings, we investigate plane structures.
In particular, we consider plane subdrawings that might appear in every simple drawing of the complete graph~$K_n$.
One of the earliest statements in that direction was the following conjecture by Rafla, which is the starting point for this paper.

\begin{conjecture}[Rafla~{\cite{Rafla1988}}]
\label{conj:rafla}
Every simple drawing of $K_n$ on $n \geq 3$ vertices contains a plane Hamiltonian cycle.
\end{conjecture}

For \emph{geometric drawings}, which correspond to point sets in the plane connected via straight-line segments, the existence of a plane Hamiltonian cycle can easily be shown.
The study of plane structures gets significantly harder in the more general setting of simple drawings.
As one of the first results on guaranteed plane structures in simple drawings of~$K_n$, Pach, Solymosi, and T{\'o}th~\cite{PachSolymosiToth2003} proved the existence of any fixed plane tree of size $\bigO((\log n)^{1/6})$.
Subsequently, plane matchings and paths were investigated as a relaxation of plane cycles \cite{pt-2010-detg, FoxSudakov2009, Suk13, Fulek2014, RuizVargas17}.
The best bounds to date are by Aichholzer, Garc{\'i}a, Tejel, Vogtenhuber, and Weinberger~\cite{agtvw-2024-twisted_journal}, who showed that every simple drawing of $K_n$ contains a plane matching of size $\Omega(n^{1/2})$ and a plane path of length $\Omega(\log n / \log \log n)$.
The latter was also independently proven by Suk and Zeng~\cite{sz-2024-upcstg}.

Rafla's conjecture itself was verified for all simple drawings of~$K_n$ on $n \leq 9$ vertices~\cite{AbregoAFHOORSV2015}.
In their paper on plane matchings, Aichholzer et al.~\cite{agtvw-2024-twisted_journal} studied so-called generalized twisted drawings of~$K_n$ and proved that they always contain a plane Hamiltonian path and, for odd~$n$, also a plane Hamiltonian cycle.

Aichholzer, Orthaber, and Vogtenhuber~\cite{aov-2024-tcfhcsdcg} recently introduced a variant of Rafla's conjecture, which asserts that simple drawings of $K_n$ are \emph{plane Hamiltonian-connected}, i.e., there exists a plane Hamiltonian path between every pair of vertices.

\begin{conjecture}[Aichholzer, Orthaber, Vogtenhuber~{\cite{aov-2024-tcfhcsdcg}}]
\label{conj:stpath}
For each pair of vertices $s,t$ in a simple drawing of $K_n$ there exists a plane Hamiltonian path from $s$ to~$t$.
\end{conjecture}

They proved \Cref{conj:rafla,conj:stpath} for the subclasses of strongly c-monotone and cylindrical drawings\footnote{%
We are not aware of any relation between these subclasses and convex drawings.
A schematic overview of known relations between various classes of simple drawings is given in \cite[Figure~7]{aov-2024-tcfhcsdcg}.},
and showed that their conjecture is indeed a strengthening:
If it was true for all simple drawings of $K_n$ it would imply that also Rafla's conjecture holds in general.

\subparagraph*{Convex drawings.}

In this article, we investigate plane Hamiltonian paths and cycles in the subclass of convex drawings.
Convex drawings were introduced by Arroyo, McQuillan, Richter and Salazar~\cite{ArroyoMRS2017_convex} as the most general class of the convexity hierarchy between geometric and simple drawings of~$K_n$.
Arroyo, Richter, and Sunohara~\cite{ArroyoRS2020} already showed Rafla's conjecture for \emph{h-convex} drawings, which are a subclass of convex drawings.
The notion of convexity in simple drawings generalizes the classic convexity and is based on \emph{triangles} of a drawing, i.e., the subdrawings induced by three vertices.
Since in a simple drawing the edges of a triangle do not cross, a triangle partitions the plane (resp.\ the sphere) into exactly two connected components.
The closures of these components are the two \emph{sides} of the triangle.
In particular, the three vertices forming the triangle lie in both sides of the triangle.
A side $S$ of a triangle is \emph{convex} if every edge that has its two incident vertices in~$S$ is fully contained inside~$S$.
A~simple drawing $\calD$ of $K_n$ is
\begin{itemize}
\item \emph{convex} if every triangle of $\calD$ has a convex side,
\item \emph{h-convex} (short for hereditarily convex) if there is a choice of a convex side $S_T$ for every triangle $T$ of $\calD$ such that, for every triangle $T'$ contained in $S_T$, it holds $S_{T'} \subseteq S_T$, and
\item \emph{f-convex} (short for face convex) if there exists a marking face~$F$ in $\calD$ such that for all triangles the side not containing $F$ is convex.
\end{itemize}
Clearly all f-convex drawings are h-convex and all those are in turn convex. 

Arroyo et al.~\cite{ArroyoMRS2017_convex} showed that convex drawings of~$K_n$ can be characterized in terms of their induced subdrawings on $5$ vertices.
To define this, we call two simple drawings of the same graph \emph{isomorphic}\footnote{%
This isomorphism is usually referred to as \emph{weak isomorphism} since there also exist stronger notions.}
if they have the same pairs of crossing edges up to relabeling of the vertices.
\Cref{fig:n5_fivetypes} shows the five non-isomorphic simple drawings of $K_5$, which are denoted as Type~I~to~V; cf.~\cite{AbregoAFHOORSV2015}.
A simple drawing is convex if and only if every induced subdrawing on $5$ vertices is isomorphic to a geometric drawing, i.e., it is of Type~I, II, or~III~\cite{ArroyoMRS2017_convex}.
The above mentioned \emph{generalized twisted drawings}, for which a plane Hamiltonian path exists \cite{agtvw-2024-twisted_journal}, are simple drawings where all 5-tuples are of Type~V~\cite{agtvw-2023-crsgtd5t}.

\begin{figure}[ht]
\centering
\includegraphics{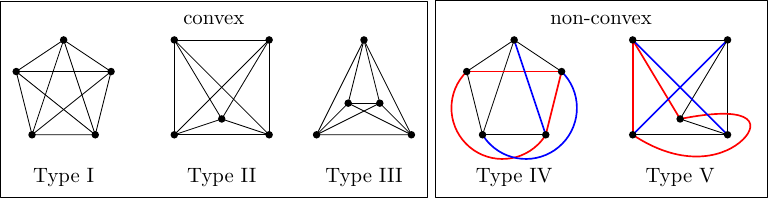}
\caption{The five non-isomorphic ways to draw the complete graph~$K_5$.
Type IV and Type V are non-convex because the red triangles have no convex side, as witnessed by the blue edges.}
\label{fig:n5_fivetypes}
\end{figure}

\subparagraph*{Maximal plane subdrawings.}

The study of plane subdrawings raises the question about the maximal number of edges in a plane subdrawing.
In this article, we investigate maximal plane subdrawings in convex drawings and show that they are in fact maximum plane.
A~\emph{maximum} plane subdrawing is a plane subdrawing with the largest number of edges among all plane subdrawings.
A plane subdrawing is called \emph{maximal} if adding any further edge from the underlying drawing would result in a crossing.
Garc\'ia, Pilz and Tejel~\cite{GarciaTejelPilz2021} showed that it is \NP-complete to determine a \emph{maximum} plane subdrawing in a simple drawing.
However, every simple drawing of $K_n$ contains a plane subdrawing with at least $2n-3$ edges~\cite[Corollary~3.4]{GarciaTejelPilz2021}, which is best possible as witnessed by the geometric drawing of $n$ points in convex position.

\subparagraph*{Empty $k$-cycles.}

Besides Hamiltonian cycles and maximum/maximal plane subdrawings, we also investigate smaller plane structures.
We introduce empty $k$-cycles as a link between empty triangles ($k=3$), which are known to exist in every simple drawing of $K_n$~\cite{Harborth98,AichholzerHPRSV2015}, and plane Hamiltonian cycles ($k=n$).
Similar as in the case of a triangle, a plane cycle partitions the plane (resp.\ sphere) into two connected components, which we call the \emph{sides} of the cycle.
A~side~$S$ is \emph{empty} if there are no vertices contained in the interior of~$S$.
An \emph{empty $k$-cycle} is a plane cycle of length $k$ such that at least one of its two sides is empty.

\subsection*{Our contribution}

In \Cref{sec:paths_and_cycles} we prove \Cref{conj:rafla,conj:stpath} for the class of convex drawings.
More specifically, in \Cref{thm:rafla_convex} we show the existence of a plane Hamiltonian cycle and in \Cref{thm:stpath_convex} the existence of a plane Hamiltonian path connecting any pair of vertices.

In \Cref{sec:maxplane} we show that in a convex drawing of $K_n$ all maximal plane subdrawings have the same number of edges (\Cref{thm:convex-maximal-is-maximum}).
This allows to find \emph{maximum} plane subdrawings in a greedy fashion and, in particular, to extend every plane Hamiltonian cycle by at least $n-3$ edges while preserving planarity (\Cref{cor:extend_to_2n_minus_3}).
Consequently, in every convex drawing of~$K_n$ there is a \emph{plane Hamiltonian subdrawing} with $2n-3$ edges, i.e., a plane subdrawing with $2n-3$ edges that contains a Hamiltonian cycle.

On top of that, we show in \Cref{sec:strengthenings} that for every convex drawing and every fixed vertex $v_\star$ there exists a plane Hamitonian subdrawing with a spanning star centered at $v_\star$ (\Cref{theorem:convex_HC}).
This structural result implies that in every convex drawing of $K_n$ every edge is contained in at least one plane Hamiltonian path (\Cref{thm:HP_prescribed_edges}).
Moreover, we obtain the existence of an empty $k$-cycle for every integer $k$ with $3 \leq k \leq n$ (\Cref{cor:emptykcycle}).

\section{Hamiltonian paths and cycles}
\label{sec:paths_and_cycles}

In this section we prove the two conjectures for all convex drawings.
In a first step we show that for each pair of vertices in a convex drawing there is a plane Hamiltonian path connecting them (\Cref{conj:stpath}).
From this we then derive the existence of a plane Hamiltonian cycle in every convex drawing (\Cref{conj:rafla}).

\begin{figure}[htb]
\centering
\subcaptionbox{\label{fig:star-cross-good}}[.49\textwidth]{\includegraphics[page=1]{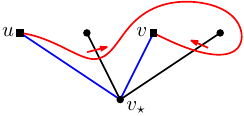}}
\subcaptionbox{\label{fig:star-cross-bad}}[.49\textwidth]{\includegraphics[page=2]{simple-star-cross-direction_dcg.pdf}}
\caption{\subref{fig:star-cross-good}~The edge $\{u,v\}$ has to cross star edges in a specific direction. \subref{fig:star-cross-bad}~Otherwise it would have to cross that star edge more than once, cross itself, or cross an adjacent star edge.}
\label{fig:star-cross-direction}
\end{figure}

For the following proofs we fix one vertex~$v_\star$, which we call the \emph{star vertex}.
By the properties of a simple drawing the edges incident to $v_\star$, which we refer to as \emph{star edges}, do not cross each other and hence form a plane spanning star.
The clockwise cyclic order of the star edges in the drawing is the \emph{rotation} of~$v_\star$.
In this context we often identify the edges with their incident vertices that are different from~$v_\star$.
Since convexity and the existence of a plane substructure are independent of the choice of the outer face of a drawing, in the following illustrations we consider $v_{\star}$ to be on the outer face.
Keep in mind that in a simple drawing of $K_n$, a non-star edge $\{u,v\}$, directed from $u$ to~$v$, has to cross star edges always from $u$ towards $v$ in the rotation of~$v_\star$; as shown in \Cref{fig:star-cross-good}.
Otherwise the edge would enter an area, marked gray in \Cref{fig:star-cross-bad}, where it cannot leave anymore without violating simplicity.
Moreover the statements in this paper are independent from the labeling of the vertices.
For convenience, we assume that $v_{\star} = n$ and that the other vertices are labeled from $1$ to $n-1$ according to the rotation of~$v_{\star}$.
To deal with the cyclic order of the vertices around $v_{\star}$, we use arithmetics modulo $n-1$.

For every non-star edge $e=\{u,v\}$, the triangle spanned by the two vertices $u$ and $v$ of $e$ and the star vertex $v_\star$ is denoted as~$T_e$.
An edge $e$ is \emph{star-crossing} if it crosses at least one of the star edges.
In particular, we focus on the edges $e = \{v, v+1\}$ with $1\leq v \leq n-1$.
We call such an edge \emph{good} if it is not star-crossing.
Otherwise, if edge $b = \{v, v+1\}$ crosses a star edge~$\{w,v_{\star}\}$, then we say that $b$ is a \emph{bad edge} with \emph{witness}~$w$; see \Cref{fig:badedge}.
Note that the side of the triangle $T_b$ containing the witness is not convex.
However, by definition every triangle in a convex drawing has at least one convex side.
Thus, the triangle $T_b$ has exactly one convex side, which is the side not containing the witness.

\begin{figure}[htb]
\centering
\includegraphics{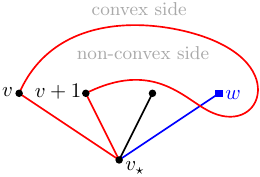}
\caption{A bad edge $b = \{v,v+1\}$ with witness~$w$ (blue). The triangle $T_b$ is highlighted in red.}
\label{fig:badedge}
\end{figure}

This shows that all witnesses are in the same side.
Furthermore, since $v$ and $v+1$ are consecutive in the rotation of~$v_\star$, all vertices in the interior of the non-convex side of $T_b$ are witnesses of~$b$.

\begin{observation}
\label{obs:ConvexSide}
Let $b=\{v,v+1\}$ be a bad edge.
Then the side of the triangle $T_b$ not containing the witnesses is the unique convex side.
Moreover, a vertex is a witness of $b$ if and only if it is in the interior of the non-convex side of~$T_b$.
\end{observation}

By the definition of convexity, an edge between two vertices in a convex side is contained in this convex side.
In the following lemma, we show that a similar property holds for the non-convex side of the triangle $T_b$ induced by a bad edge~$b$.

\begin{lemma}
\label{lemma:edges_nonconvexside}
Let $b$ be a bad edge and let $e$ be a non-star edge connecting two vertices from the non-convex side of~$T_b$.
Then $e$ does not cross the triangle $T_b$ and is therefore contained in the non-convex side of~$T_b$.
\end{lemma}

\begin{proof}
Clearly the edge $e = b$ is contained in the non-convex side of $T_b$.

Next, we consider an edge~$e$ from a vertex $w$ in the interior of the non-convex side of~$T_b$ to a vertex of~$b$.
By \Cref{obs:ConvexSide}, $w$ is a witness.
Hence, in the subdrawing induced by the four vertices $v_\star,v,v+1,w$ the bad edge $b$ and the star edge~$\{w,v_\star$\} cross.
Since a simple drawing of $K_4$ has at most one crossing, $e$ cannot cross the triangle~$T_b$.

Finally, we consider an edge $e = \{ w, w' \}$ connecting two interior vertices of the non-convex side of $T_b$, i.e., $e$ connects two witnesses of~$b$.
We assume, without loss of generality, that $v,v+1,w,w'$ appear in exactly this order in the rotation of~$v_\star$, because $v$ and $v+1$ are consecutive vertices in that rotation.

Assume towards a contradiction that $e$ crosses~$T_b$.
Since both vertices of $e$ lie in the same side of~$T_b$, the edge $e$ crosses $T_b$ an even number of times.
This implies that $e$ crosses exactly two of the three edges of $T_b$ and, in particular, at least one of its star edges.
By symmetry, we assume that $e$ crosses $\{v+1,v_\star\}$.
Then $e$ passes through the region bounded by $\{v+1,v_\star\}$, and by the parts of $b$ and $\{w,v_\star \}$ up to their crossing; see \Cref{fig:non-convex-side_non-simple}.
Consequently $e$ crosses $b$, which fixes all crossings of $e$ with~$T_{b}$.

Now consider the triangle $T_{e}$.
The side of~$T_{e}$ containing $v+1$ is not convex because $\{v+1,v_\star\}$ crosses~$e$.
Since $b$ crosses all three edges of~$T_{e}$, $v$ lies in the unique convex side of~$T_{e}$.
Therefore $\{w,v\}$ lies in this convex side.
Starting at $w$, $\{w,v\}$ passes through the region bounded by $\{w,v_\star \}$, and by parts of $e$ and $\{v+1,v_\star\}$, leaving through $\{v+1,v_\star\}$; see \Cref{fig:non-convex-side_non-convex2}.
Also $b$ passes through this region, crossing $e$ and $\{w,v_\star \}$.
Thus, $\{w,v\}$ crosses its adjacent edge $b$; a contradiction.
This shows that in all cases $e$ cannot cross $T_b$ and hence stays in the non-convex~side.
\end{proof}

\begin{figure}[htb]
\centering 
\subcaptionbox{\label{fig:non-convex-side_non-simple}}[.49\textwidth]{\includegraphics[page=3]{newlemma_dcg.pdf}}
\subcaptionbox{\label{fig:non-convex-side_non-convex2}}[.49\textwidth]{\includegraphics[page=1]{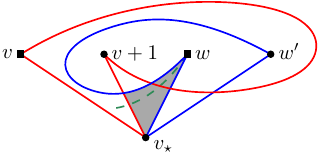}}
\caption{\subref{fig:non-convex-side_non-simple}~If an edge $e = \{w,w'\}$ (dotted blue) between two witnesses crosses a star edge of~$T_{b}$, then it crosses~$b$.
\subref{fig:non-convex-side_non-convex2}~Moreover, in this case the edge $\{w,v\}$ (dashed green) crosses its adjacent edge~$b$.}
\label{fig:non-convex-side}
\end{figure}

\Cref{lemma:edges_nonconvexside} is the key tool to show \Cref{conj:stpath} for convex drawings.
It will also be important that for two non-star edges $e$ and $e'$ that do not cross star edges, the rotation of~$v_\star$ determines whether $e$ and $e'$ cross.

\begin{observation}
\label{obs:HCplane}
Let $e = \{u,v\}$ and $e' = \{u',v'\}$ be two non-star edges that do not cross any star edge.
If $u,v,u',v'$ appear in this cyclic order around~$v_\star$, then $e$ and $e'$ do not cross.
\end{observation}

The main idea to show the existence of a Hamiltonian path between each pair of vertices is to divide the drawing into two parts, one part consisting of the vertices in the convex side of a triangle $T_b$ and the other one with the vertices from the non-convex~side.

\begin{theorem}
\label{thm:stpath_convex}
For each pair of vertices $s,t$ in a convex drawing of $K_n$ there exists a plane Hamiltonian path from $s$ to~$t$.
\end{theorem}

\begin{proof}
We prove this statement by induction on~$n$.
The base cases $n=2$ and $n=3$ are trivial.
For the inductive step, we pick $t$ as the star vertex~$v_\star$.
If there is no bad edge, we start at $s$, traverse all other vertices in cyclic order around~$t$ via good edges, and finally take a star edge to terminate in~$t$; see \Cref{fig:stpath_good}.
This Hamiltonian path is plane by \Cref{obs:HCplane}.

Now assume that there is a bad edge $b=\{v,v+1\}$.
The triangle $T_b$ partitions the remaining vertices into two parts:
The vertices $V_C$ in the interior of its convex side~$S_C$, which might be an empty set, and the vertices $V_N$ in the interior of its non-convex side~$S_N$, which contains at least one witness of~$b$.
Note that $V_C$ and $V_N$ do not contain vertices of the triangle $T_b$.
By convexity, for each pair of vertices from~$V_C \cup \{v,v+1,t\}$ the connecting edge is contained in~$S_C$.
By \Cref{lemma:edges_nonconvexside}, for each pair of vertices from~$V_N \cup \{v,v+1\}$ the connecting edge is contained in~$S_N$. In particular, those edges do not cross the edges of the triangle $T_b$.
We distinguish the following two cases depending on the position of~$s$.

\begin{description}
\item[Case 1: $s \in V_C$ ($s$ is in the interior of the convex side).]
By induction there exists a plane path~$P_1$ in the subdrawing induced by $V_C \cup \{v\}$ from $s$ to~$v$, which traverses all vertices in~$V_C \cup \{v\}$ and is contained in~$S_C$.
Similarly, by induction there exists a plane path~$P_2$ in the subdrawing induced by $V_N \cup \{v,v+1\}$ from $v$ to~$v+1$, which traverses all vertices in~$V_N \cup \{v,v+1\}$ and is contained in~$S_N$.
The concatenation of~$P_1$,~$P_2$, and the star edge~$\{v+1,t\}$ yields the desired plane Hamiltonian path from $s$ to~$t$; see \Cref{fig:stpath_convex_case2}.

\item[Case 2: $s \not\in V_C$ ($s$ is in the non-convex side).]
By symmetry, we assume without loss of generality that $s \not=v$.
By induction there exists a plane path $P_1$ in the subdrawing induced by $V_N \cup \{v,v+1\}$ from $s$ to~$v$, which traverses all vertices of $V_N \cup \{v,v+1\}$ and is contained in~$S_N$.
Note that this in particular includes the case $s = v+1$.
Similarly, by induction there exists a plane path $P_2$ in the subdrawing induced by $V_C \cup \{t,v\}$ from $v$ to~$t$, which traverses all vertices of $V_C \cup \{t,v\}$ and is contained in~$S_C$.
The concatenation of these two paths~$P_1$ and $P_2$ yields the desired plane Hamiltonian path from $s$ to~$t$; see \Cref{fig:stpath_convex_case1}. \qedhere
\end{description}
\end{proof}

\begin{figure}[htb]
\centering
\subcaptionbox{\label{fig:stpath_good}}[.34\textwidth]{\includegraphics[page=1]{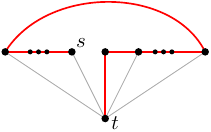}}
\subcaptionbox{\label{fig:stpath_convex_case2}}[.32\textwidth]{\includegraphics[page=2]{stpath_dcg.pdf}}
\subcaptionbox{\label{fig:stpath_convex_case1}}[.32\textwidth]{\includegraphics[page=4]{stpath_dcg.pdf}}
\caption{\subref{fig:stpath_good}~illustrates a plane Hamiltonian $s$-$t$-path with good edges. \subref{fig:stpath_convex_case2}~and \subref{fig:stpath_convex_case1}~illustrate the two cases to construct a plane Hamiltonian $s$-$t$-path if a bad edge exists.}
\label{fig:stpath_convex}
\end{figure}

\begin{theorem}
\label{thm:rafla_convex}
Every convex drawing of $K_n$ with $n \geq 3$ contains a plane Hamiltonian~cycle.
\end{theorem}

\begin{proof}
We choose an arbitrary vertex $t$ as the star vertex.
If there is no bad edge, we can easily construct a plane Hamiltonian cycle by traversing all but one good edges according to the rotation of $t$ and close that path with two consecutive star edges.

Otherwise, we pick a bad edge $b = \{v,v+1\}$ and construct a Hamiltonian path from $s = v+1$ to $t$ as in Case~2 of the proof of \Cref{thm:stpath_convex}.
The path consists of two parts, $P_1$~and~$P_2$, which exist by applying \Cref{thm:stpath_convex} to the vertices in the non-convex side and the convex side of~$T_b$, respectively.
Path $P_1$ connects $s$ to $v$ and is contained in the non-convex side~$S_N$, and $P_2$ connects $v$ to $t$ and is contained in the convex side~$S_C$.
Since none of the path edges crosses the star edge $\{t,s\}$, closing the path with the edge $\{t,s\}$ yields the desired plane Hamiltonian~cycle; see \Cref{fig:cycle_convex}.
\end{proof}

\begin{figure}[htb]
\centering
\includegraphics[page=3]{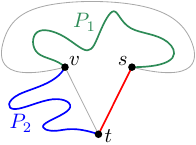}
\caption{Constructing a plane Hamiltonian cycle in a convex drawing.}
\label{fig:cycle_convex}
\end{figure}

\section{Maximal plane subdrawings}
\label{sec:maxplane}

In the previous section we have seen that convex drawings of $K_n$ contain a plane Hamiltonian cycle.
Moreover, for each pair of vertices there is a plane Hamiltonian path connecting them.
In this section we investigate how many edges we can add to a Hamiltonian cycle in a convex drawing while staying plane.
In other words, we ask for maximal plane Hamiltonian substructures.
As a first step, we omit the Hamiltonicity and show that all maximal plane subdrawings of a convex drawing have the same number of edges.
This is very different from the general case of simple drawings where it is \NP-complete to decide whether there exists a plane subdrawing with a given number of edges~\cite{GarciaTejelPilz2021}.

\begin{theorem}
\label{thm:convex-maximal-is-maximum}
Every maximal plane subdrawing of a convex drawing $\calD$ of~$K_n$ is maximum plane.
\end{theorem}

We use the following three results for our proof.
Garc{\'i}a, Pilz, and Tejel~\cite{GarciaTejelPilz2021} studied properties of maximal plane subdrawings of general simple drawings.

\begin{lemma}[{Garc{\'i}a, Pilz, Tejel~\cite[Theorem~3.1]{GarciaTejelPilz2021}}]
\label{lem:maxplane_spanning2conn}
In every simple drawing $\calD$ of $K_n$, every maximal plane subdrawing of $\calD$ is spanning and 2-connected.
\end{lemma}

For a plane cycle $C = (v_1, \ldots, v_k)$ in a simple drawing $\calD$, a \emph{chord} of $C$ is an edge of $\calD$ connecting two non-consecutive vertices $v_i$ and $v_j$ of~$C$.
A chord of $C$ can either cross some edges of $C$ or it lies entirely in one of the two sides of~$C$.

\begin{lemma}[{Garc{\'i}a, Pilz, Tejel~\cite[Lemma~3.6]{GarciaTejelPilz2021}}]
\label{lem:chordsincomplement}
Let $\calD$ be a simple drawing of $K_n$ and $C = (v_1, \ldots, v_k)$, $k \geq 3$, a plane cycle in $\calD$ with sides $S_1$ and $S_2$.
If no chord of $C$ lies entirely in $S_1$, then all chords of $C$ lie entirely in $S_2$. 
\end{lemma}

A subdrawing induced by the vertices $v_1, \ldots, v_k$ of a plane cycle~$C$ in a simple drawing of~$K_n$ such that all chords of~$C$ lie on the same side of~$C$ is isomorphic to a geometric drawing of $k$ points in convex position, denoted by~$\calC_k$.
Moreover, we call the side of~$C$ containing all chords the \emph{convex side} of~$\calC_k$.
In the introductory paper of Arroyo et al.~\cite{ArroyoMRS2017_convex} for convex drawings, they studied properties of vertices of the convex side of $\calC_k$ and the behaviour of edges between two such vertices (which includes the vertices of~$\calC_k$).

\begin{lemma}[{Arroyo, McQuillan, Richter, Salazar~\cite[Lemma~3.5]{ArroyoMRS2017_convex}}]
\label{lem:convex}
Let $\calD$ be a convex drawing of $K_n$ with a subdrawing $\calC_k$, $k \geq 4$.
Then every edge of $\calD$ between two vertices of the convex side of $\calC_k$ stays in the convex side of~$\calC_k$.
\end{lemma}

We are now ready to prove the theorem.

\begin{proof}[Proof of \Cref{thm:convex-maximal-is-maximum}]
Let $\calD'$ be an arbitrary maximal plane subdrawing of~$\calD$.
Then $\calD'$ is $2$-connected and spanning by \Cref{lem:maxplane_spanning2conn}.
Clearly, in a $2$-connected plane drawing, every face is bounded by a cycle.
We show that, for every $k \geq 4$, the $k$-faces are identical in every maximal plane subdrawing of $\calD$, where a face is a \emph{$k$-face} if it is bounded by exactly $k$ edges.

Let $F$ be a $k$-face in $\calD'$ with $k \geq 4$ and denote its bounding cycle by~$C$, i.e., $F$ is one of the two sides of~$C$.
Since $\calD'$ is maximal plane, no chord of $C$ lies entirely in~$F$.
Consequently, \Cref{lem:chordsincomplement} implies that all chords of~$C$ lie entirely on the other side $S_2$ of~$C$.
In particular, the vertices of $C$ induce a subdrawing~$\calC_k$ of~$\calD$ isomorphic to the geometric drawing of $k$ points in convex position.
Moreover, since~$\calD'$ is spanning and $F$ is a face of~$\calD'$, no vertex of~$\calD$ lies inside of~$F$.
Since $\calD$ is convex, we use \Cref{lem:convex} to conclude that all edges of~$\calD$ are in the convex side $S_2$ of~$\calC_k$.
Thus, $C$ is completely uncrossed in~$\calD$, and $C$ and $F$ are part of every maximal plane subdrawing of~$\calD$.

This fixes, for every $k \geq 4$ and every maximal plane subdrawing $\calD'$ of~$\calD$, the number of $k$-faces in $\calD'$.
Hence, Euler's formula implies $f_3 = e + c$, where $f_3$ is the number of 3-faces, $e$~is the number of edges, and $c$ is a constant depending on the number of vertices and faces of size at least four in~$\calD'$.
On the other hand, double counting the number of edges in $\calD'$ via incident faces gives $2e = 3f_3 + c'$ for some constant $c'$ depending on the number of faces of size at least four.
Since there is a unique solution to these two equations, the number of 3-faces and, especially, the number of edges in every maximal plane subdrawing is fixed.
Thus, every maximal plane subdrawing of~$\calD$ is maximum~plane.
\end{proof}

Recall the specific structure of maximal plane subdrawings from the proof.
In some sense it is similar to triangulations of a pointgon (see Aichholzer, Rote, Speckmann, and Streinu~\cite{arss-2003-zppt} for a formal definition) with holes, i.e., triangulations of a polygon with holes and additional points inside.
Further note that a drawing of $K_n$ is f-convex if and only if it is pseudolinear~\cite{ArroyoMRS2017_pseudolines}.
Thus, every f-convex drawing has a convex hull and every maximal plane subdrawing of an f-convex drawing is a triangulation of that convex hull.

In the setting of simple drawings and in accordance with Garc{\'i}a, Pilz, and Tejel~\cite{GarciaTejelPilz2021}, we call a plane drawing on $n$ vertices with $3n-6$ edges a \emph{triangulation}.
Equivalently, a triangulation is a plane drawing where every face is a triangle.
For h-convex drawings we can prove the existence of such a triangulation.

\begin{theorem}
\label{thm:h-convex-maximal-plane}
Let $\calD$ be an h-convex drawing of $K_n$.
Then in every maximal plane subdrawing of $\calD$ all but at most one faces are triangles.
Moreover, if $\calD$ is not f-convex, then every maximal plane subdrawing of $\calD$ is a triangulation.
\end{theorem}

To prove the latter part, we again use a lemma from Arroyo et al.~\cite{ArroyoMRS2017_convex}.

\begin{lemma}[{Arroyo, McQuillan, Richter, Salazar~\cite[Lemma~4.7]{ArroyoMRS2017_convex}}]
\label{lem:h-convex}
Let $\calD$ be an h-convex drawing of $K_n$ with a subdrawing $\calC_k$, $k \geq 4$, such that $\calD$ lies entirely in the convex side of~$\calC_k$.
Then $\calD$ is f-convex.
\end{lemma}

\begin{proof}[Proof of \Cref{thm:h-convex-maximal-plane}]
Let $\calD'$ be an arbitrary maximal plane subdrawing of~$\calD$.
Assume towards a contradiction that there are at least two faces $F_1$ and $F_2$ in $\calD'$ with bounding cycles $C_{k_1}$ and $C_{k_2}$, respectively, of lengths $k_1, k_2 \geq 4$.
As shown in the proof of \Cref{thm:convex-maximal-is-maximum}, $C_{k_1}$~and $C_{k_2}$ are completely uncrossed in $\calD$ and they induce subdrawings $\calC_{k_1}$ and $\calC_{k_2}$ isomorphic to geometric drawings in convex position, which both contain $\calD$ in their convex sides.
Hence, $\calC_{k_2}$ lies in the convex side of $\calC_{k_1}$ and vice versa.
Moreover, every triangle induced by three vertices of $\calC_{k_i}$, $i \in \{ 1, 2 \}$, has a unique convex side, which is contained in the convex side of~$\calC_{k_i}$.
Since $C_{k_2}$ is completely uncrossed, this implies that there is a triangle~$T_1$ induced by three vertices of $\calC_{k_1}$ such that $C_{k_2}$ and consequently $\calC_{k_2}$ lies completely in the unique convex side of~$T_1$.
Similarly, there is a triangle $T_2$ induced by three vertices of $\calC_{k_2}$ such that $\calC_{k_1}$ lies completely in the unique convex side of~$T_2$.
But then $T_1$ and $T_2$ form a contradiction to h-convexity; this shows the first part.

For the second part, assume that there is one face $F$ in $\calD'$ that is not a triangle.
Again, the vertices on the boundary of $F$ induce a $\calC_k$, $k \geq 4$, such that $\calD$ is contained in the convex side of~$\calC_k$.
Hence, by \Cref{lem:h-convex}, $\mathcal{D}$ is f-convex.
Consequently, every maximal plane subdrawing of an h-convex but not f-convex drawing is a triangulation.
\end{proof}

Arroyo, Richter, and Sunohara~\cite{ArroyoRS2020} suggested that all crossing minimizing drawings of~$K_n$ are h-convex.
For the following conjecture, we are inspired by two additional facts.
First, all crossing minimizing geometric and pseudolinear drawings of $K_n$ have a triangular convex hull~\cite{agor-2007-nlbnkercn,blprs-2007-cheopdt}.
Second, no pseudolinear drawing of $K_n$, for $n$ large enough, is crossing minimizing~\cite{acfls-2012-kechlgdkn}.

\begin{conjecture}
Every maximal plane subdrawing of every crossing minimizing drawing of $K_n$, $n\geq 3$, is a triangulation.
\end{conjecture}

We come back to convex drawings.
Garc{\'i}a, Pilz, and Tejel~\cite[Corollary~3.4]{GarciaTejelPilz2021} showed that every simple drawing of $K_n$ has a plane subdrawing with at least $2n-3$ edges.
Together with \Cref{thm:convex-maximal-is-maximum} this shows that all maximal plane subdrawings of a convex drawing of $K_n$ have at least $2n-3$ edges.
In general, this bound is best possible as every triangulation of $n$ points in convex position has exactly $2n-3$~edges.

\begin{corollary}
\label{cor:convex-plane-edge-count}
Every maximal plane subdrawing of a convex drawing of~$K_n$ has at least $2n-3$ edges.
\end{corollary}

In contrast, for general simple drawings it is known that all maximal plane subdrawings have at least $\frac{3n}{2}$ edges, which is tight as shown by Garc{\'i}a, Pilz, and Tejel~\cite[Proposition~3.3]{GarciaTejelPilz2021}.
They give a construction of a (non-convex) simple drawing $\calD$ of $K_{n}$ having a maximal plane subdrawing $\calD'$ with $\left\lceil\frac{3n}{2}\right\rceil$ edges.
Actually, $\calD'$ has a Hamiltonian cycle that crosses all edges of $\calD \setminus \calD'$.
This shows that there are plane Hamiltonian cycles in simple drawings that cannot be extended by more than $\left\lceil\frac{n}{2}\right\rceil$ edges while staying plane.
We repeat their construction for completeness:
Place $n$ vertices $v_1, \ldots, v_n$ clockwise at the vertices of a regular $n$-gon such that $v_1$ is the leftmost vertex and the line between $v_2$ and $v_n$ is vertical.
Draw the cycle $C = v_1, v_2, \ldots, v_n$ straight-line and draw all chords of $C$ straight-line inside of $C$ except for those that would not cross any vertical chord.
Draw them and the edge between $v_1$ and $v_{\frac{n}{2}+1}$ outside of~$C$; see \Cref{fig:planecyclewith32n} for an illustration with $n=10$ vertices.
The Hamiltonian cycle in question is then $C_H = v_1, v_2, v_{n}, v_{n-1}, v_3, \ldots, v_{\frac{n}{2}+1}$.
For odd $n$, both edges from $v_1$ to $v_{\left\lfloor\frac{n}{2}\right\rfloor+1}$ and $v_{\left\lceil\frac{n}{2}\right\rceil+1}$ are drawn outside of $C$ and exactly one of them is part of~$C_H$.

\begin{figure}[htb]
\centering
\includegraphics{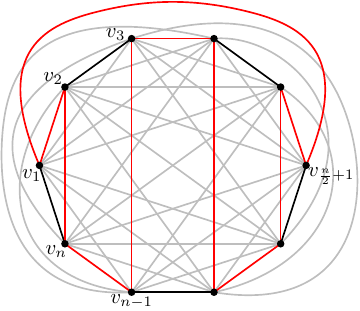}
\caption{A drawing of $K_n$ with a Hamiltonian cycle (red) that crosses all but $\left\lceil\frac{n}{2}\right\rceil$ many other edges (black).}
\label{fig:planecyclewith32n}
\end{figure}

For convex drawings, combining \Cref{thm:rafla_convex,thm:convex-maximal-is-maximum}, we obtain that every plane Hamiltonian cycle can be extended by $n-3$ edges in a greedy fashion to obtain a plane Hamiltonian subdrawing with $2n-3$ edges.

\begin{corollary}
\label{cor:extend_to_2n_minus_3}
Every plane Hamiltonian cycle in a convex drawing of~$K_n$ can be greedily extended to a plane Hamiltonian subdrawing on $2n-3$ edges.
\end{corollary}

\pagebreak

\section{Stars and cycles}
\label{sec:strengthenings}

As we have seen in the last section, in convex drawings we can extend every plane Hamiltonian cycle to a plane subdrawing with $2n-3$ edges.
In this section we investigate the structure of such plane Hamiltonian subdrawings.
In particular, we show that for every vertex $v_\star$ there exists a plane Hamiltonian subdrawing with $2n-3$ edges that contains all edges incident to~$v_\star$.
We defer the proof to~\Cref{sec:convex_HC_proof}.

\begin{theorem}
\label{theorem:convex_HC}
For every convex drawing $\calD$ of $K_n$ with $n \geq 3$ and every vertex $v_\star$ in~$\calD$, there exists a plane Hamiltonian cycle that does not cross any edge incident to~$v_\star$.
Such a Hamiltonian cycle can be computed in $\bigO(n^2)$ time.
\end{theorem}

This theorem provides a lot of structure for plane Hamiltonian subdrawings, which makes it possible to investigate further properties.
For example, already in geometric drawings there are edges that are not contained in any plane Hamiltonian cycle (see one of the red edges in \Cref{fig:noprescribededges_a}).
However, replacing this property with Hamiltonian paths, we show that the plane Hamiltonian paths of a given convex drawing cover all~edges.

\begin{corollary}
\label{thm:HP_prescribed_edges}
For every convex drawing $\calD$ of $K_n$ and every edge $e$ in $\calD$, there exists a plane Hamiltonian path containing~$e$.
\end{corollary}

\begin{proof}
Let $e = \{u,v\}$ be an arbitrary edge of~$\calD$.
We consider $u$ as the star vertex.
By \Cref{theorem:convex_HC} there exists a plane Hamiltonian subdrawing $\calD'$ of $\calD$ that contains all edges incident to~$u$.
Let the plane Hamiltonian cycle $C$ in $\calD'$ traverse $u,x_1,\ldots,x_{n-1}$ in this order with $v=x_i$ for some index~$i$.
If $i=1$ or $i=n-1$, then $C$ contains the edge $\{u,v\}$ and therefore yields a Hamiltonian path that fulfills the desired property.
Otherwise $x_{i-1},x_{i-2},\ldots,x_{1},u,v,x_{i+1},\ldots,x_{n-1}$ is a Hamiltonian path containing the edge~$e$.
It is plane because all its edges belong to the plane drawing~$\calD'$.
\end{proof}

Given this property of prescribing an edge for a Hamiltonian path, a natural generalization is to prescribe more than one edge.
To obtain a plane substructure, the prescribed edges cannot induce a crossing in the drawing.
Moreover, as illustrated in \Cref{fig:noprescribededges_a}, prescribing two adjacent edges or longer subpaths is not possible in general.
It is however possible to prescribe two non-crossing independent edges for a plane Hamiltonian path in a geometric or even pseudolinear drawing of~$K_n$.

\begin{figure}[htb]
\centering
\subcaptionbox{\label{fig:noprescribededges_a}}[.4\textwidth]{\includegraphics[page=1]{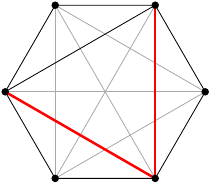}}
\subcaptionbox{\label{fig:noprescribededges_b}}[.48\textwidth]{\includegraphics[page=2]{two_prescribed_dcg.pdf}}
\caption{\subref{fig:noprescribededges_a}~Two adjacent edges in a geometric drawing of $K_6$ that cannot be extended to a plane Hamiltonian path. 
\subref{fig:noprescribededges_b}~An illustration of \Cref{proposition:HP_pair_indep_edges}: extending two independent edges in a pseudolinear drawing of~$K_n$ to a plane Hamiltonian path.}
\label{fig:noprescribededges}
\end{figure}

\begin{proposition}
\label{proposition:HP_pair_indep_edges}
For every pseudolinear drawing $\calD$ of $K_n$ and every pair of non-crossing independent edges $e$ and~$e'$ in $\calD$, there exists a plane Hamiltonian path containing $e$ and~$e'$.
\end{proposition}

\begin{proof}
Let $e = \{u,v\}$ and $e'= \{u',v'\}$ be the two independent edges that do not cross in the pseudolinear drawing~$\calD$.
By the definition of pseudolinear drawings, we can extend both $e$ and $e'$ to pseudolines $\ell$ and $\ell'$, respectively.
And because $e$ and $e'$ do not cross, we may assume that $e$ has no crossing with $\ell'$ as otherwise we exchange the roles of $e$ and~$e'$.
So $\ell'$ partitions the plane into two regions with $e$ being contained in one of them.
We further partition the region containing $e$ using the pseudoline $\ell$ and obtain three regions: $R_1$ and $R_2$ contain $e$ on its boundary, and $R_3$ has $e'$ on its boundary; see \Cref{fig:noprescribededges_b} for an illustration.
Moreover, we may assume $u'$ lies on the boundary of $R_2$ as otherwise we exchange the roles of $R_1$ and~$R_2$.
By starting at an arbitrary interior vertex $x$ of $R_1$, we find a plane path $P_1$ in $R_1$ that traverses all interior vertices from $R_1$ and ends in vertex~$u$.
For this we either start at $x$ and traverse the points in the rotation around $x$ starting at the point after~$u$.
Or, if $x$ lies on the convex hull of the vertices in $R_1$ without $v$, we traverse the points on one side of the pseudoline through $x$ and $u$ in the rotation around $x$ until the point $w$ before $u$ and then traverse the remaining vertices in the rotation around~$w$.
Note that this is the statement of \Cref{thm:stpath_convex} for pseudolinear drawings.
Similarly, there is a plane path $P_2$ in $R_2$ starting at vertex~$v$, traversing all interior vertices from~$R_2$, and ending in vertex~$u'$.
In the same manner, there is a plane path $P_3$ in $R_3$ starting at vertex~$v'$ and traversing all interior vertices from~$R_3$.
By concatenating $P_1$, $e$, $P_2$, $e'$, and $P_3$ we obtain the desired Hamiltonian path, which is plane because the three regions are separated by the pseudolines $\ell'$ and~$\ell$.
\end{proof}

Even though we only used a basic notion of left and right in the argument above, the result does not generalize to convex drawings because we cannot extend edges in that setting.
And indeed prescribing two edges is in general not possible in convex drawings.

\begin{lemma}
For every $n \geq 8$, there exists an h-convex drawing $\mathcal{D}$ of $K_n$ and non-crossing independent edges $e$ and $e'$ in $\mathcal{D}$ such that no plane Hamiltonian path contains $e$ and $e'$.
\end{lemma}

\begin{proof}
We define a specific drawing $\mathcal{D}$ of $K_n$ for given $n \geq 8$ by placing the vertices $1, \ldots, n$ in that order on a horizontal line and drawing each edge either \emph{above} or \emph{below} that line.
We fix the edges $e = \{ a, c \}$ and $e' = \{ b, n \}$ with $1 < a < b-1$ and $b+1 < c < n-1$.
Moreover, we define the vertex sets $A = \{ 1, \ldots, a \}$, $B = \{ a+1, \ldots, b \}$, $C = \{ b+1, \ldots, c \}$, and $D = \{ c+1, \ldots, n \}$.
Then we draw all edges between vertices of $B$ and $C$, between vertices of $A$ and $D$, and between vertices of $B$ and $D$ below.
All other edges we draw above.
In particular, $e$ is drawn above and $e'$ is drawn below.
\Cref{fig:not2prescribededges} shows the resulting drawing for $n=8$.

We next show, by contradiction, that $e$ and $e'$ are not contained in any common plane Hamiltonian path~$P$.
Such a potential path $P$ consist of three parts, one before $e$, one between $e$ and~$e'$, and one after~$e'$.
Let $A' = A \setminus \{ a \}$ and similarly $B'$, $C'$, and $D' = D \setminus \{ n \}$.
Since all edges between two of those sets are crossed either by $e$ or~$e'$, each of the three parts of $P$ can contain vertices of at most one of the four sets $A'$, $B'$, $C'$, and~$D'$.
And because all four sets are non-empty but there are only three parts of $P$ to fit them, such a plane Hamiltonian path $P$ cannot exist.

To see that $\mathcal{D}$ is h-convex, we consider all triangles in the drawing.
In particular, for all triangles with exactly one vertex in $A$, $C$, and $D$ their outside is convex.
For all other triangles their inside is convex.
And it is a straightforward task to check that this choice of convex sides fulfills the properties of an h-convex drawing, which finishes the proof.
\end{proof}

\begin{figure}[htb]
\centering
\includegraphics[page=2]{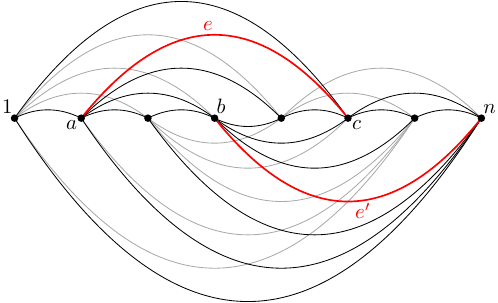}
\caption{An h-convex drawing of $K_8$. The two prescribed edges, highlighted red, are not contained in any common plane Hamiltonian path.}
\label{fig:not2prescribededges}
\end{figure}

Note that the constructed drawing $\mathcal{D}$ is actually almost f-convex, i.e., pseudolinear:
If either $A$, $C$, or $D$ contains only $2$ vertices, then removing those $2$ vertices makes the insides of all triangles convex.

Besides prescribing edges, \Cref{theorem:convex_HC} implies that for every $3 \leq k \leq n$ every convex drawing admits a plane cycle of length $k$, as we argue below.
In abstract graphs the existence of cycles of arbitrary length is often referred to as \emph{pancyclicity}.
Recall that a plane cycle has two sides.
In addition to cycles of arbitrary length, we find cycles such that one side does not contain vertices in its interior, i.e., \emph{empty $k$-cycles}.
For $k=3$ those are known as \emph{empty triangles}, a concept which is considered often in the study of simple drawings.
For simple drawings of $K_n$ with $n\ge 3$, Harborth~\cite{Harborth98} proved that there are at least two empty triangles and conjectured that the minimum among all simple drawings of~$K_n$ is $2n-4$.
Today it is known that there are at least $n$ empty triangles~\cite{AichholzerHPRSV2015}, and that all generalized twisted drawings have exactly $2n-4$ empty triangles~\cite{gtvw-2023-etgtdkn}.

Note that every plane Hamiltonian cycle is an empty $n$-cycle as both of its sides are empty.
Moreover, from the properties of the plane Hamiltonian cycle constructed in \Cref{theorem:convex_HC}, we derive the existence of empty $k$-cycles for all $k$.

\begin{theorem}\label{thm:empty-k-cycles}
If a simple drawing $\calD$ of $K_n$ contains a vertex $v_\star$ and a plane Hamiltonian cycle that does not cross any edge incident to $v_\star$, then there exists an empty $k$-cycle in $\calD$ for every integer $3 \leq k \leq n$.
\end{theorem}

\begin{proof}
The plane Hamiltonian cycle consists of two star edges incident to~$v_\star$ and a plane Hamiltonian path $P$ on the remaining vertices $1, \ldots, n-1$.
The crucial property is that $P$ does not cross any star edge.
The plane path $P$ starts at some vertex $v \in \{1, \ldots, n-1\}$.
Similar to \Cref{obs:HCplane}, since $P$ visits all vertices $1, \ldots, n-1$ and does not cross itself, for every $1 \le k \le n-1$, the indices of the first $k$ vertices of $P$ determine an integer interval $\{i,i+1,\ldots,j-1,j\}$.
Hence the next vertex $v'$ of $P$ is either $i-1$ or $j+1$.
Inductively it follows that closing the path after visiting $k-1$ vertices with the two star edges $\{v,v_\star\}$ and $\{v',v_\star\}$ gives the desired empty $k$-cycle.
\end{proof}

\begin{corollary}
\label{cor:emptykcycle}
For every convex drawing $\calD$ of $K_n$ and every integer $3 \leq k \leq n$ there exists an empty $k$-cycle in~$\calD$.
\end{corollary}

\pagebreak

\section{Plane Hamiltonian subdrawings in convex drawings}
\label{sec:convex_HC_proof}

In this section we give the full proof of \Cref{theorem:convex_HC}.
We prove the existence of the Hamiltonian cycle in a constructive way and present an $\bigO(n^2)$-time algorithm that, for a given convex drawing $\calD$ of $K_n$ and a fixed vertex $v_{\star}$ in $\calD$, computes a plane Hamiltonian cycle that does not cross edges incident to $v_{\star}$.

As in the proof of Rafla's conjecture for convex drawings (\Cref{thm:rafla_convex}) we consider~$v_\star$ as the star vertex, assume $v_\star = n$ and label the remaining vertices with $1, \ldots, n-1$ corresponding to the rotation of~$v_\star$.
Similar to \Cref{sec:paths_and_cycles}, bad edges will play an important role in this proof.
For drawings with at most one bad edge, we can concatenate $n-2$ good edges around~$v_\star$ to obtain a plane path through all non-star vertices, which does not cross any star edges.
To obtain the desired plane Hamiltonian cycle, we add the two star edges connecting the two end-vertices of the path to~$v_\star$.

\begin{observation}
\label{lem:onebadedgePHC}
If there is at most one bad edge, then $\calD$ contains a plane Hamiltonian cycle that does not cross any star edges and visits the non-star vertices in the cyclic order around the star vertex $v_{\star}$.
\end{observation}

Hence in the following we consider the case that there are at least two bad edges.
Recall that, for a bad edge $b$, by \Cref{obs:ConvexSide} the unique convex side of $T_b$ is the side not containing the witnesses.
To deal with multiple bad edges, we investigate their structure in the following lemma, which extends \Cref{lemma:edges_nonconvexside}.

\begin{lemma}
\label{lem:twobadedges1}
Let $b = \{v, v+1\}$ and $b' = \{v', v'+1\}$ be two distinct bad edges with witnesses $w$ and $w'$, respectively. Then the following four statements hold:
\begin{enumerate}[(i)]
\item \label{lem:twobadedges1:item0} Either the edge $b$ does not cross any star edge of~$T_{b'}$ or the edge $b'$ does not cross any star edge of~$T_b$.
\item \label{lem:twobadedges1:item1} The triangle $T_b$ is contained in the convex side of $T_{b'}$ and vice versa.
\item \label{lem:twobadedges1:item2} $w \neq w'$.
\item \label{lem:twobadedges1:item3} $w',w,v,v'$ appear in this or the reversed cyclic order in the rotation of~$v_{\star}$.
\end{enumerate}
\end{lemma}

\begin{proof}
To prove Property~\eqref{lem:twobadedges1:item0}, we first consider the case that $b$ and $b'$ share a vertex.
We assume without loss of generality that $v+1 = v'$.
The subdrawing induced by the four vertices $v,v',v'+1,v_\star$ is a simple drawing of $K_4$ and has therefore at most one crossing.
If both $b$ and $b'$ cross a star edge of the other triangle, there would be two crossings on that~$K_4$.
Hence the property holds in this case.

For the remaining case in which $b$ and $b'$ are independent, consider the subdrawing induced by the five distinct vertices $v,v+1,v',v'+1,v_\star$.
If $b$ crosses both star edges of $T_{b'}$, then $v'$ and $v'+1$ are both in the non-convex side of~$T_b$.
Thus, by \Cref{lemma:edges_nonconvexside}, $b'$ does not cross any star edge of $T_b$ and the property follows.
An equivalent argument holds if $b'$ crosses both star edges of~$T_b$.

Hence, from now on we assume that $b$ crosses exactly one star edge of $T_{b'}$ and $b'$ crosses exactly one star edge of~$T_b$.
Without loss of generality assume that $b$ crosses $\{v',v_\star\}$, the other case is symmetric.
Since $b'$ crosses exactly one star edge of $T_b$ and the two vertices $v'$ and $v'+1$ are on different sides of $T_b$, the edge $b'$ crosses $T_b$ exactly once.
Similarly to the arguments in the proof of \Cref{lemma:edges_nonconvexside}, a crossing of $b'$ and $\{v+1,v_\star\}$ forces an additional crossing with~$b$.
For an illustration see \Cref{fig:onecrossing_a_app}.
Hence, $b'$~crosses the star edge $\{v,v_\star\}$.

We now consider the edge $\{v+1,v'+1\}$.
Its vertices $v+1$ and $v'+1$ are both in the convex side of both triangles $T_b$ and~$T_{b'}$.
To stay in the convex side of $T_{b}$, starting at $v+1$, the edge $\{v+1,v'+1\}$ is forced to cross $\{v',v_\star\}$, see \Cref{fig:onecrossing_b_app}.
This contradicts $\{v+1,v'+1\}$ staying in the convex side of $T_{b'}$, which completes the proof of Property~\eqref{lem:twobadedges1:item0}.

\begin{figure}[htb]
\centering  
\subcaptionbox{\label{fig:onecrossing_a_app}}[.49\textwidth]{\includegraphics[page=4]{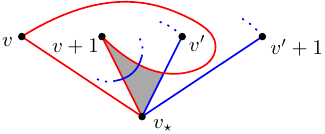}}
\subcaptionbox{\label{fig:onecrossing_b_app}}[.49\textwidth]{\includegraphics[page=2]{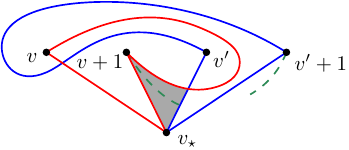}}
\caption{Illustrations of the two case where $b$ and $b'$ cross exactly one star edge of the other triangle: \subref{fig:onecrossing_a_app}~$b'$ crosses $\{v+1,v_\star\}$ and \subref{fig:onecrossing_b_app}~$b'$ crosses $\{v,v_\star\}$.}
\label{fig:onecrossing_app}
\end{figure}

Next we derive Property~\eqref{lem:twobadedges1:item1} from Property~\eqref{lem:twobadedges1:item0}.
Assume without loss of generality that $b$ does not cross any star edge of $T_{b'}$.
Then, by definition, $v'$ and $v'+1$ are not witnesses for~$b$ and hence by \Cref{obs:ConvexSide} lie in the convex side of~$T_b$.
Since $v_\star$ lies in the convex side of $T_b$, by convexity, the triangle $T_{b'}$ is contained in the convex side of $T_b$.
In particular, $b'$~does not cross any star edge of $T_b$.
With the same arguments as for $b$, $T_b$ is contained in the convex side of $T_{b'}$, implying Property~\eqref{lem:twobadedges1:item1}.

To show Property~\eqref{lem:twobadedges1:item2}, recall that, by \Cref{obs:ConvexSide}, the witnesses $w$ and $w'$ lie in the interior of the non-convex side of $T_b$ and $T_{b'}$, respectively.
Now Property~\eqref{lem:twobadedges1:item1} implies that these non-convex sides are interiorly disjoint, and therefore $b$ and $b'$ do not have any common witness, i.e., $w \neq w'$.

It remains to show Property~\eqref{lem:twobadedges1:item3}, which concerns the cyclic order of vertices around~$v_\star$.
Up to symmetries there are three possibilities, namely $w,w',v,v'$ or $w,v,w',v'$ or the claimed case $w',w,v,v'$.
In the first case, it is not possible that $b'$ crosses $\{w',v_\star\}$ without crossing $T_b$ or $\{w,v_\star\}$ because $\{w',v_\star\}$ is contained in a region bounded by $\{v,v_\star\}$ and parts of $\{w,v_\star\}$ and~$b$.
However, crossing $T_b$ would contradict Property~\eqref{lem:twobadedges1:item1} and crossing $\{w,v_\star\}$ would contradict Property~\eqref{lem:twobadedges1:item2}.
For an illustration see \Cref{fig:cyclicorder_a_app}.

In the second case, we consider the edge $\{v,v'\}$; see \Cref{fig:cyclicorder_b_app}. 
Since $v$ is in the convex side of $T_{b'}$ and $v'$ is in the convex side of $T_b$, $\{v,v'\}$ is in the intersection of both convex sides.
In particular, it crosses $\{w,v_\star\}$ and $\{w',v_\star\}$.
Moreover, the triangle $T_{\{v,v'\}}$ (spanned by the edge $\{v,v'\}$ and $v_\star$) has $w$ on one side and $w'$ on the other.
Hence $T_{\{v,v'\}}$ has no convex side; a contradiction.
This completes the proof of Property~\eqref{lem:twobadedges1:item3}.
\end{proof}

\begin{figure}[htb]
\centering
\subcaptionbox{\label{fig:cyclicorder_a_app}}[.49\textwidth]{\includegraphics[page=1]{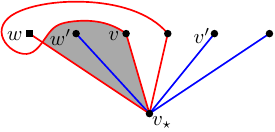}}
\subcaptionbox{\label{fig:cyclicorder_b_app}}[.49\textwidth]{\includegraphics[page=2]{cyclicorder_dcg.pdf}}
\caption{Illustration that the two shown cyclic orders cannot appear in a convex drawing: \subref{fig:cyclicorder_a_app}~the cyclic order $w,w',v,v'$ and \subref{fig:cyclicorder_b_app} the cyclic order $w,v,w',v'$.}
\label{fig:cyclicorder_app}
\end{figure}

Property~\eqref{lem:twobadedges1:item3} implies that the vertices $1, \ldots, n-1$ in the rotation of $v_{\star}$ can be partitioned into two blocks of consecutive vertices such that one block contains all vertices of all bad edges and the other block contains all the witnesses.
In particular, if $b$ is the bad edge whose vertices are last in the clockwise order of its block, we can cyclically relabel the vertices such that $b$ becomes $\{n-2,n-1\}$.
This makes the labels of all witnesses smaller than the labels of vertices of bad edges.
We have the following two properties:

\begin{description}
\item[(sidedness)] If $\{v,v+1\}$ is a bad edge with witness $w$, then $w < v$.
\item[(nestedness)] If $b = \{v, v+1\}$ and $b' = \{v', v'+1\}$ are bad edges with respective witnesses $w$ and $w'$ and if $v < v'$, then $w > w'$.
\end{description}

In addition we can choose the outer face such that the vertex~$v_{\star}$ and the initial parts of the edges $\{v_\star,1\}$ and $\{v_\star,n-1\}$ are incident to it.

The nesting property implies that we can label the bad edges as $b_1, \ldots, b_m$ for some $m \geq 2$, such that if $b_i = \{v_i,v_i +1\}$, then $1 < v_1 < v_2 < \ldots < v_m = n-2$.
Moreover, let $w_i^L$ and $w_i^R$ denote the leftmost (smallest index) and the rightmost (largest index) witness of the bad edge $b_i$, respectively.
Then $1 \leq w_{m}^L \leq w_{m}^R < w_{m-1}^L \leq w_{m-1}^R < \ldots < w_{1}^L \leq w_{1}^R $.
Sidedness additionally implies $w_{i}^R < v_i$ for all $i =1, \ldots, m$.
\Cref{fig:layeringbadedgeswithnotation} shows the situation for two bad edges $b_i$ and $b_{i+1}$.
Note that $v_i+1 = v_{i+1}$ is possible.

\begin{figure}[htb]
\centering
\includegraphics{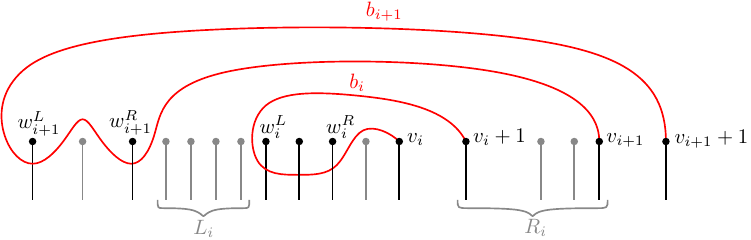}
\caption{Illustration of sidedness and nestedness for two bad edges. All vertical edges are incident to the star vertex $v_\star$.}
\label{fig:layeringbadedgeswithnotation}
\end{figure}

In a first step, we consider edges between one of the two vertices of a bad edge $b_i$ and a witness of $b_i$.
In particular, we show that $\{w_i^L, v_i+1\}, \{w_i^R, v_i\}$ are not star-crossing.

\begin{lemma}
\label{lem:leftrightmostwitness}
For all $i = 1, \ldots, m$, neither $\{w_i^L, v_i+1\}$ nor $\{w_i^R, v_i\}$ is star-crossing.
\end{lemma}

\begin{proof}
For every fixed $i$, \Cref{lemma:edges_nonconvexside} shows that both edges are contained in the non-convex side of the triangle $T_{b_i}$.
Assume $\{w_i^L,v_i+1\}$ crosses a star edge $\{x,v_{\star}\}$.
Then $x$ is a witness of $b_i$ with $w_i^L < x$ by the choice of~$w_i^L$.
However, the side of the triangle $\{w_i^L,v_i+1,v_{\star}\}$ that contains $x$ is not convex due to the edge $\{x,v_{\star}\}$.
Additionally, the other side is not convex due to the edge $\{v_i,v_i+1\}$; a contradiction.
A similar argument holds for the edge $\{w_i^R,v_i\}$. The two situations are depicted in \Cref{fig:left-most-witness,fig:right-most-witness}.
\end{proof}

\begin{figure}[htb]
\centering
\subcaptionbox{\label{fig:left-most-witness}}[.49\textwidth]{\includegraphics[page=2]{insidebadedge_dcg.pdf}}
\subcaptionbox{\label{fig:right-most-witness}}[.49\textwidth]{\includegraphics[page=1]{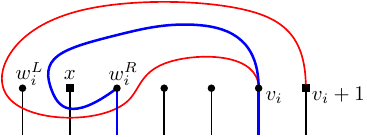}}
\caption{The edges $\{w_i^L,v_i+1\}$ and $\{w_i^R,v_i\}$ cannot be star-crossing because otherwise the blue triangle has no convex side, as witnessed by edges incident to the vertices marked with squares. \subref{fig:left-most-witness} depicts the first case and \subref{fig:right-most-witness} the second case.}
\label{fig:left-right-most-witness}
\end{figure}

For $i = 1, \ldots, m-1$ let $L_i = \{\, x \mid w_{i+1}^R < x < w_i^L \,\}$ and $R_i = \{\, x \mid v_{i} +1\leq x \leq v_{i+1} \,\}$ denote the left and the right blocks of vertices between two consecutive bad edges $b_i$ and~$b_{i+1}$; see \Cref{fig:layeringbadedgeswithnotation}.
Note that $R_i$ is non-empty since it always contains $v_i+1$ and $v_{i+1}$ but $L_i$ might be empty.

In the following we seek edges from $L_i$ to $R_i$ that are not star-crossing.
First, we show that the properties of a convex drawing imply that the edges between vertices of $L_i$ and $R_i$ stay in the region between the two bad edges $b_i$ and $b_{i+1}$.
This implies that the only star edges which can be crossed by those edges are $\{x,v_\star\}$ with $x \in L_i \cup R_i$.

\begin{lemma}
\label{lem:nocrossingsoutsidelandR}
No edge $\{u,v\}$ with $u,v \in L_i \cup R_i$ crosses star edges $\{z,v_{\star}\}$ with $z \in V \setminus (L_i \cup R_i)$.
\end{lemma}

\begin{proof}
The convex sides $S_i$ and $S_{i+1}$ of the triangles $T_{b_i}$ and $T_{b_{i+1}}$, respectively, have a common intersection which is partitioned into three regions by the edges $\{w_{i+1}^R,v_{\star}\}$ and $\{w_i^L,v_{\star}\}$.
Both vertices $u,v$ are contained in the region that is bounded by both edges $\{w_{i+1}^R,v_{\star}\}$ and $\{w_i^L,v_{\star}\}$; shaded gray in \Cref{fig:stay-between-bad-edges}.
Since the edge $\{u,v\}$ lies in $S_i$ and $S_{i+1}$ and crosses $\{w_{i+1}^R,v_{\star}\}$ and $\{w_i^L,v_{\star}\}$ at most once, it is contained in the same region.
This shows that we cannot cross star edges $\{z,v_{\star}\}$ where $z$ is outside of this region, i.e., $z \in V \setminus (L_i \cup R_i)$.
\end{proof}

\begin{figure}[htb]
\centering
\includegraphics[page=1]{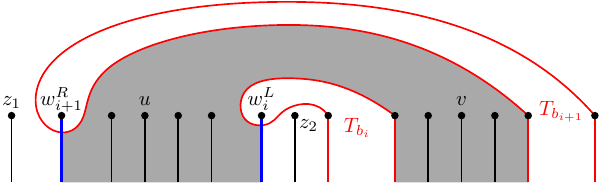}
\caption{Illustration for the proof of \Cref{lem:nocrossingsoutsidelandR}.
All edges $\{u,v\}$ with both end-vertices in the gray shaded region are contained in that region.}
\label{fig:stay-between-bad-edges}
\end{figure}

\begin{figure}[htb]
\centering
\subcaptionbox{\label{fig:nocrossinginRiCase1}}[.49\textwidth]{\includegraphics[page=1]{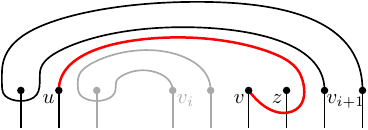}}
\subcaptionbox{\label{fig:nocrossinginRiCase2}}[.49\textwidth]{\includegraphics[page=2]{notexistingconfiguration_dcg.pdf}}
\caption{Illustration of the two forbidden configurations \subref{fig:nocrossinginRiCase1} and \subref{fig:nocrossinginRiCase2} to prove \Cref{lem:nocrossinginRi}.
The red edges cannot cross the star edge $\{z,v_{\star}\}$ as depicted.
The vertices and incident edges that are drawn gray are deleted to achieve a subdrawing contradicting the convexity.}
\label{fig:nocrossinginRi}
\end{figure}

Moreover, we show that an edge from a vertex in $L_i$ to a vertex in $R_i$ cannot cross star edges with vertices in~$R_i$.

\begin{lemma}
\label{lem:nocrossinginRi}
No edge $\{u,v\}$ with $u \in L_i$ and $v \in R_i$ crosses star edges $\{z,v_{\star}\}$ with $z \in R_i$.
\end{lemma}

\begin{proof}
Assume towards a contradiction that the edge $\{u,v\}$ crosses a star edge $\{z,v_{\star}\}$ with $z \in R_i$.
Let $w_i$ and $w_{i+1}$ be witnesses of the bad edges $b_i = \{v_i,v_i+1\}$ and $b_{i+1} = \{v_{i+1}, v_{i+1}+1\}$, respectively.
We distinguish the two cases of $z > v$ and $z < v$, which are depicted in \Cref{fig:nocrossinginRiCase1,fig:nocrossinginRiCase2}.
For $z>v$, we consider the subdrawing $\calD_1$ induced by the 7 distinct vertices $w_{i+1},u,v,z,v_{i+1}, v_{i+1}+1,v_{\star}$, where $v_i+1=v$ is possible.
In $\calD_1$ the vertices $u$ and $v$ are consecutive in the rotation of $v_{\star}$, thus, $\{u,v\}$ and $b_{i+1}$ are both bad edges.
However, the respective cyclic order of the vertices around $v_{\star}$ violates Property~\eqref{lem:twobadedges1:item3} of \Cref{lem:twobadedges1}; a contradiction.
For $z<v$, we consider the subdrawing $\calD_2$ induced by $u,w_i,v_i,v_i+1,z,v,v_{\star}$, where $v=v_{i+1}$ is possible.
Again, since $u$ and $v$ are consecutive in the cyclic order around~$v_\star$, the edge $\{u,v\}$ is a bad edge with witness $z$ in $\calD_2$ and the order of vertices corresponding to the two bad edges violates Property~\eqref{lem:twobadedges1:item3} of \Cref{lem:twobadedges1}.
\end{proof}

The previous lemmas show that if $\{u, v\}$ with $u\in L_i$ and $v\in R_i$ crosses a star edge $\{z,v_{\star}\}$, then $z\in L_i$.
We now analyze crossings in $L_i$ and show that the edge $\{u,v\}$ cannot cross two star edges $\{z_1,v_{\star}\}$ and $\{z_2,v_{\star}\}$ with $z_1 <u <z_2$ and $z_1,z_2 \in L_i$; see \Cref{fig:crossingsonlyoneside}.

\begin{lemma}
\label{lem:crossingsonlyoneside}
No edge $\{u,v\}$ with $u \in L_i$ and $v \in R_i$ crosses two star edges $\{z_1,v_{\star}\}$ and $\{z_2,v_{\star}\}$ with $z_1 <u < z_2$ and $z_1,z_2 \in L_i$.
\end{lemma}

\begin{proof}
Assume towards a contradiction that $\{u,v\}$ crosses both edges $\{z_1,v_{\star}\}$ and $\{z_2,v_{\star}\}$.
Then $z_1$ and $z_2$ are on different sides of the triangle spanned by $\{u,v,v_{\star}\}$.
Since both star edges $\{z_1,v_{\star}\}$ and $\{z_2,v_{\star}\}$ cross the triangle, there is no convex side; a contradiction.
\end{proof}

\begin{figure}[htb]
\centering
\subcaptionbox{\label{fig:crossingsonlyonesideCase1}}[.49\textwidth]{\includegraphics[page = 1]{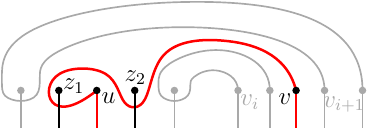}}
\subcaptionbox{\label{fig:crossingsonlyonesideCase2}}[.49\textwidth]{\includegraphics[page = 2]{edgescrossed3_dcg.pdf}}
\caption{Illustration for the proof of \Cref{lem:crossingsonlyoneside}.
In both cases \subref{fig:crossingsonlyonesideCase1} and \subref{fig:crossingsonlyonesideCase2} the red triangle $\{u,v,v_{\star}\}$ has no convex side.
Witnesses for the non-convexity are the edges $\{z_1,v_{\star}\}$ and $\{z_2,v_{\star }\}$.}
\label{fig:crossingsonlyoneside}
\end{figure}

Even though edges from $L_i$ to $R_i$ cannot cross star edges incident to vertices in $L_i$ with smaller and larger indices at the same time, crossings on one side cannot be avoided.
We focus on edges which only cross star edges with larger indices than the end-vertex in $L_i$.
As the following lemma shows, those edges help us to find edges from $L_i$ to $R_i$ that are not star-crossing.

\begin{lemma}
\label{lem:edgesfromlefttoright}
Let $u \in L_i$ and $v \in R_i$ and let $z$ be the largest index in $L_i$ with $z>u$ such that the edge $\{u,v\}$ crosses the star edge $\{z,v_{\star}\}$.
Then the following two statements~hold:
\begin{enumerate}
\item \label{lem:edgesfromlefttoright:item1} The edge $\{z,v\}$ is not star-crossing.
\item \label{lem:edgesfromlefttoright:item2} The edge $\{z+1,v\}$ does not cross any star edge $\{x',v_{\star}\}$ with $x' \in L_i$ and $x' \le z$.
\end{enumerate}
\end{lemma}

\begin{proof}
To show~\eqref{lem:edgesfromlefttoright:item1}, assume towards a contradiction that the edge $\{z,v\}$ crosses a star edge $\{x',v_{\star}\}$.
From \Cref{lem:nocrossingsoutsidelandR,lem:nocrossinginRi} we know that $x' \in L_i$.
Moreover, because a simple drawing of $K_{4}$ has at most one crossing and the edges $\{u,v\}$ and $\{z,v_{\star}\}$ already cross, the edge $\{z,v\}$ has no crossing with the triangle $T_{\{u,v\}}$.
In particular, $\{z,v\}$ cannot cross any star edge in the convex side of $T_{\{u,v\}}$, i.e., the side not containing~$z$.
The choice of $z$ implies that all edges $\{x,v_{\star}\}$ with $x \in L_i$ and $x >z$ do not cross $\{u,v\}$.
Hence $\{z,v\}$ can only cross star edges $\{x',v_{\star}\}$ with $x' \in L_i$ and $x' < z$; \Cref{fig:edgesfromlefttoright} shows two examples with $x' = x'_1 < u$ and $x' = x'_2 > u$.
However, then the triangle $T_{\{z,v\}}$ has no convex side:
The side containing the vertex~$x'$ is not convex since the edge $\{x',v_{\star}\}$ crosses $\{z,v\}$.
The other side contains the vertex~$u$ and is not convex since the edge $\{u,v\}$ crosses $\{z,v_{\star}\}$; a contradiction.

To show~\eqref{lem:edgesfromlefttoright:item2}, assume towards a contradiction that $\{z+1,v\}$ crosses a star edge $\{x',v_{\star}\}$ with $x' \in L_i$ and $x' \leq z$.
Since $\{z+1,v\}$ cannot cross its adjacent edges $\{u,v\}$ and $\{v,v_{\star}\}$, it crosses first $\{z,v_{\star}\}$ and later $\{u,v_{\star}\}$ in order to cross any star edge $\{x',v_{\star}\}$ with $x' \leq z$.
Note that this also holds if $z+1 = w_i^L$.
An illustration is given in \Cref{fig:edgesfromlefttoright2}.
Now consider the triangle $T = \{z,z+1,v\}$ and the star edge $\{u,v_{\star}\}$, which crosses $\{z+1,v\}$.
As a star edge, $\{u,v_{\star}\}$ is crossed neither by $\{z,z+1\}$, which is a good edge since $z+1 < v_1$, nor by $\{z,v\}$, which follows from part~\eqref{lem:edgesfromlefttoright:item1}.
Thus $\{u,v_{\star}\}$ crosses exactly one of the edges of $T$ and, consequently, $u$ and $v_\star$ lie in different sides of~$T$.
However none of those two sides is convex:
The edge $\{u,v\}$ crosses $\{z,v_{\star}\}$ but not $\{z+1,v_{\star}\}$.
Hence it crosses the good edge $\{z,z+1\}$ which shows that the side of $T$ containing $u$ is not convex.
The side of $T$ containing $v_\star$ is not convex because $\{z,v_{\star}\}$ crosses $\{z+1,v\}$.
This is a contradiction to the convexity.
\end{proof}

\begin{figure}[htb]
\centering
\includegraphics{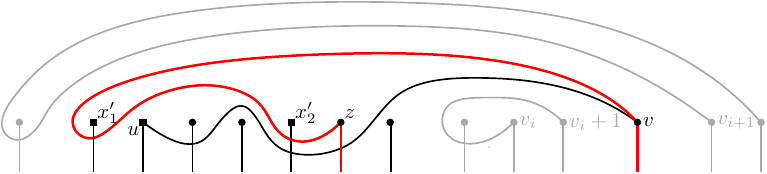}
\caption{Illustration for the proof of \Cref{lem:edgesfromlefttoright}\eqref{lem:edgesfromlefttoright:item1}. The red triangle $T_{\{z,v\}}$ has no convex side. Witnesses for the non-convexity are the edges $\{u,v\}$ and $\{x_i',v_{\star }\}$.}
\label{fig:edgesfromlefttoright}
\end{figure}

\begin{figure}[htb]
\centering
\includegraphics{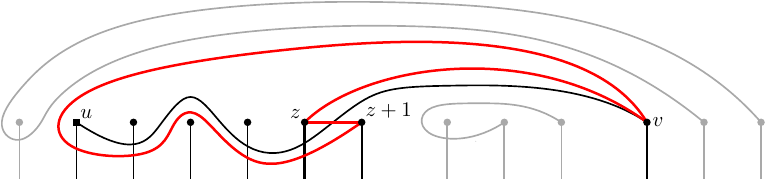}
\caption{Illustration for the proof of \Cref{lem:edgesfromlefttoright}\eqref{lem:edgesfromlefttoright:item2}. The red triangle has no convex side. Witnesses for the non-convexity are the edges $\{z,v_\star\}$ and $\{u,v\}$.}
\label{fig:edgesfromlefttoright2}
\end{figure}

We need one more definition to construct the plane Hamiltonian cycle.
Consider a fixed vertex $r \in R_i$ and all edges $\{u,r\}$ with $u \in L_i$ that cross a star edge $\{z,v_{\star}\}$ with $z>u$.
We are interested in the rightmost such vertex $z$, if it exists.
More formally, we define
\begin{align*}
l(r) = \max ( \{\, z \in L_i \mid \exists\, u \in L_i \text{ , } u<z : \{u,r\} \text{ crosses } \{z,v_{\star}\} \,\} \cup \{w^R_{i+1}\} ).
\end{align*}
Note that in the case where $l(r)$ is set to $w^R_{i+1}$, there is no edge $\{u,r\}$ that crosses a star edge to the right of $u$.
This especially implies that the edge $\{w_{i+1}^R +1,r\}$ is not star-crossing.
For most cases, we get two non-star-crossing edges for each vertex $r \in R_i$.

\begin{lemma}
\label{cor:edgeslefttoright}
For all $r \in R_i$: If $l(r) > w^R_{i+1}$, then the edge $\{l(r),r\}$ is not star-crossing. If $l(r)+1 < w^L_{i}$, then the edge $\{l(r)+1,r\}$ is not star-crossing.
\end{lemma}

\begin{proof}
The edge $\{l(r),r\}$ is not star-crossing by \Cref{lem:edgesfromlefttoright}\eqref{lem:edgesfromlefttoright:item1} if $l(r) \in L_i$, i.e., $l(r) > w^R_{i+1}$.
Moreover, if $l(r) > w^R_{i+1}$ \Cref{lem:edgesfromlefttoright}\eqref{lem:edgesfromlefttoright:item2} implies that $\{l(r)+1,r\}$ does not cross any star edge $\{x,v_{\star}\}$ with $x \leq l(r)$.
By the maximality of $l(r)$, the edge $\{l(r)+1,r\}$ does not cross any star edge $\{x,v_{\star}\}$ with $x > l(r)+1$ if $l(r)+1 \in L_i$.
The condition $l(r)+1 \in L_i$ is fulfilled exactly if $l(r)+1 < w^L_{i}$.
In the case in which $l(r) = w^R_{i+1}$ only appears if there is no edge $\{u,r\}$ crossing a star edge right of~$u$.
Hence the edge $\{l(r)+1,r\}$ is not star-crossing due to the definition of~$l(r)$.
\end{proof}

Together with \Cref{lem:leftrightmostwitness} this gives us our desired non-star-crossing edges from $L_i$ to~$R_i$.
To make sure that these edges do not cross each other, we show that the $l(r)$ have decreasing indices with increasing~$r$.

\begin{lemma}
\label{lem:l(vi)}
For $r,r' \in R_i$ with $r<r'$, we have $l(r) \geq l(r')$.
\end{lemma}

\begin{proof}
We consider the case $r' = r+1$.
For $r' > r+1$ the claim then follows by transitivity.
If $l(r+1) = w^R_{i+1}$, the claim clearly holds.
In the following, let $l(r+1) > w^R_{i+1}$ and $u<l(r+1)$ such that $\{u,r+1\}$ crosses $\{l(r+1),v_{\star}\}$.
Consider the triangle $T_{\{u,r+1\}}$, which is drawn blue in \Cref{fig:lefttoright}.
The side of $T_{\{u,r+1\}}$ containing $l(r+1)$ is not convex, which is witnessed by $\{l(r+1),v_{\star}\}$.
By the convexity of the drawing, the side not containing $l(r+1)$ is convex.
By \Cref{lem:nocrossinginRi}, the edge $\{u,r+1\}$ does not cross $\{r,v_{\star}\}$.
Hence the vertex $r$ and consequently the edge $\{u,r\}$ are contained in the unique convex side of $T_{\{u,r+1\}}$.
This implies that $\{u,r\}$ crosses the star edge $\{ l(r+1),v_{\star}\}$, which shows $l(r) \geq l(r+1)$.
\end{proof}

\begin{figure}[htb]
\centering
\includegraphics{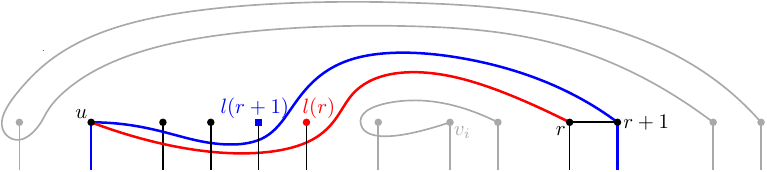}
\caption{Illustration of the proof of \Cref{lem:l(vi)}.}
\label{fig:lefttoright}
\end{figure}

The monotonicity of the vertices $l(r)$ implies the following claim, which we need for the construction of the Hamiltonian cycle.

\begin{observation}
\label{obs:remainingcases}
Let $r \in R_i$.
\begin{itemize}
\item If $l(r) = w^R_{i+1}$ for $r < v_{i+1}$, then for all $r' \in R_i$ with $r' >r$ it is $l(r') = l(r)$.
\item If $l(r)+1 = w^L_{i}$ for $r > v_{i}+1$, then for all $r' \in R_i$ with $r' < r$ it is $l(r') = l(r)$.
\end{itemize}
\end{observation}

\begin{proof}[Proof of \Cref{theorem:convex_HC}]
The general idea to construct a plane Hamiltonian cycle that is not star-crossing is now as follows:
Starting at $v_1$ (the leftmost vertex incident to a bad edge), in each step we add the next unvisited vertex to the left (smaller index) or to the right (larger index) using non-star-crossing edges.
This yields a path through all vertices except~$v_\star$, which is plane by \Cref{obs:HCplane}.
By adding the two star edges incident to the end-vertices of this path, we then obtain the desired Hamiltonian cycle.

To determine which non-star-crossing edges are suitable for our path,
we relate vertices from the right block $R = \bigcup R_i = \{v_1+1,\ldots,n-2\}$ with vertices from the left block $L = \{1,\ldots,v_1\}$.
As we have shown, for each $i=1,\ldots,m-1$ and each vertex $r \in R_i $ 
there is a corresponding vertex $l(r) \in L_i \cup \{w_{i+1}^R\}$.
If $w_{i+1}^R < l(r) < w_i^L -1$ the two edges $\{l(r),r\}$ and $\{l(r)+1,r\}$ are not star-crossing (cf.~\Cref{cor:edgeslefttoright}).
In the two remaining cases, i.e., if $l(r) = w_{i+1}^R$ or $l(r) +1 = w_i^L$, at least one of the two edges is not star-crossing with the additional property that the next, respectively previous, vertex $r'$ in $R_i$ has the same value $l(r') = l(r)$ (cf.~\Cref{obs:remainingcases}).
Moreover, for increasing index of~$r \in R$ the values $l(r) \in L$ are decreasing (cf.~\Cref{lem:l(vi)}).

Using this notion, we explicitly describe the Hamiltonian cycle:
We initialize $x := v_1$ and $r := v_1+1$.
By repeating the following procedure, we iteratively extend the plane path visiting all vertices $\{x,\ldots,r-1\}$ to a plane path visiting all vertices of $\{x', \ldots, r'-1\}$ with $x'<x$ and $r'>r$.

We set $x' := l(r)$ and $r' := \min ( \{\, \tilde{r} \in R \mid r < \tilde{r}, l(r) \neq l(\tilde{r}) \,\} \cup \{n-1\} )$, i.e., $r'$ is the first vertex to the right of $r$ where $l(r') \neq l(r)$.
Note that for $r \in R_i$ with $r< v_{i+1}$, $r' \in R_i$ and hence the edges connecting two consecutive vertices in the rotation of $v_\star$ between $r$ and $r'$ are good edges.
We traverse $x,x-1,\ldots,x'+1$ via good edges, use the non-star-crossing edge $\{x'+1,r\} = \{l(r)+1,r\}$ to reach~$r$, traverse $r,r+1,\ldots,r'-1$ via good edges, and  use the non-star-crossing edge $\{r'-1,x'\} = \{r'-1,l(r'-1)\}$ to reach~$x'$.
Note the edge $\{x'+1,r\} = \{l(r)+1,r\}$ is not star-crossing for $x' + 1 = l(r)+1 < w_{i}^L$ by \Cref{cor:edgeslefttoright}.
If $x' +1 = l(r)+1 = w_{i}^L$, then by definition of $r$ and \Cref{obs:remainingcases} it is $r =v_i+1$.
Hence the edge $\{x'+1,r\} = \{l(r)+1,r\}$ is not star-crossing by \Cref{lem:leftrightmostwitness}.
Similarly, the edge $\{r'-1,x'\} =\{r'-1,l(r'-1)\} $ is not star-crossing.
This follows for $l(r'-1) > w_{i+1}^R$ from \Cref{cor:edgeslefttoright}.
In the remaining case, if $x' = l(r'-1) = w_{i+1}^R$, then by definition of $r'$ and \Cref{obs:remainingcases} it is $r'-1 =v_{i+1}$.
Hence the edge $\{r'-1,x'\} =\{r'-1,l(r'-1)\} $ is not star-crossing by \Cref{lem:leftrightmostwitness}.

We repeat this step with $x'$ and $r'$ in the roles of $x$ and~$r$, respectively, until we have included vertex $n-2$ to the path.
Then we traverse all remaining non-star vertices $x',{x'-1},\ldots,1,n-1$ in this order via good edges, and close the Hamiltonian cycle via the two star edges $\{n-1,v_\star\}$ and $\{v_\star,v_1\}$.

\subparagraph*{Running time:}

In a first preprocessing step, we compute all bad edges.
This is possible in $\bigO(n^2)$ time:
for each of the $n-1$ edges $\{i,i+1\}$ there are at most $n-3$ potential witnesses to test.
This also determines the number $m$ of bad edges and all values $v_i$, $w_i^L$, and~$w_i^R$.

In a second preprocessing step, we compute the value $l(r)$ for every $r$.
We claim that this can be done in $\bigO(n^2)$ time.
To determine $l(r)$ for every $r \in R_i$, recall that $l(r+1) \leq l(r)$ due to~\Cref{lem:l(vi)}.
Thus, we start with the smallest index $r \in R_i$ and with the largest index $l \in L_i$.
We then iteratively check if $l = l(r)$ by testing whether $\{l,v_{\star}\}$ crosses some edge $\{l',r\}$ with $l' < l$.
If the answer is yes, we have found $l = l(r)$ and increase $r$ by one, otherwise $l \neq l(r)$ and we decrease $l$ by one.
In total, we consider a linear number of candidate pairs $(l,r)$ in this process to determine all values $l(r)$.
And for each such pair $(l,r)$, we check a linear number of potential crossings $\{l,v_{\star}\}$ with $\{l',r\}$.
Altogether, this computes all values $l(r)$ in $\bigO(n^2)$ time.

Observe that after this preprocessing we can decide in constant time, which edge to add next to the path in each step of the algorithm.
And since the final cycle has exactly $n$ edges, building it takes $\bigO(n)$ time.
Hence, the total running time is $\bigO(n^2)$, which completes the proof of \Cref{theorem:convex_HC}.
\end{proof}

\section{Conclusion}
\label{sec:discussion}

In this article, we investigated Rafla's conjecture on plane Hamiltonian cycles and variations of it in the restricted setting of convex drawings.
Since our proofs use specific properties, which do not necessarily apply in the more general setting of simple drawings, not all shown statements are true for simple drawings.

A non-convex simple drawing for which some structural results do not hold is the Type~V, which is also known as the \emph{twisted drawing} of~$K_5$ and depicted in \Cref{fig:n5_fivetypes}.
In general, the twisted drawing of $K_n$ contains exactly one edge $e$ that crosses all its independent edges.
Hence $e$ cannot be contained in any plane Hamiltonian path (cf.~\Cref{thm:HP_prescribed_edges}).
Starting with a star at one end-vertex of $e$, there is no possibility to obtain a Hamiltonian path on the remaining vertices not crossing the star edges.
Moreover, the drawings of $K_6$ and $K_7$ depicted in \Cref{fig:nostar+HC} have not a single plane subdrawing consisting of a spanning star and a plane Hamiltonian~cycle (cf.~\Cref{theorem:convex_HC}); for $n \geq 8$ again the twisted drawing of $K_n$ has that property.

\begin{figure}[htb]
\centering
\subcaptionbox{\label{fig:nostar+HC-K6}}[.49\textwidth]{\includegraphics[page=1]{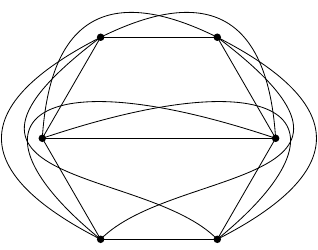}}
\subcaptionbox{\label{fig:nostar+HC-K7}}[.49\textwidth]{\includegraphics[page=2]{n6_no-star+hc_dcg.pdf}}
\caption{Non-convex drawings of \subref{fig:nostar+HC-K6}~$K_6$ and \subref{fig:nostar+HC-K7}~$K_7$ in which no vertex admits a plane Hamiltonian cycle not crossing the star edges.}
\label{fig:nostar+HC}
\end{figure}

Similar to the proof of \Cref{thm:convex-maximal-is-maximum} it can be shown that every maximal plane subdrawing of a convex drawing $\calD$ of $K_n$ contains at least $n-2$ triangular faces.
Thus, $\calD$ contains a family of at least $n-2$ pairwise disjoint empty triangles.
This is very different from simple drawings in general, for example, it follows from the definition of generalized twisted drawings~\cite{agtvw-2024-twisted_journal} that no three odd faces in them can be pairwise disjoint.
On the other hand, Zeng~\cite{z-2025-ndfstg} recently showed that every simple drawing of $K_n$ contains a linear size family of pairwise disjoint (not necessarily empty) $4$-faces, and conjectured this to be true for all even $k$-faces.

It remains open whether Rafla's conjecture (\Cref{conj:rafla}) and \Cref{conj:stpath} hold in simple drawings.
Moreover, the computational data obtained by Bergold and Scheucher~\cite{bs-2025-isdknsat-arxiv}\footnote{%
In an earlier version~\cite{BFRS2023} we combined the computational experiments with a computer-assisted proof of \Cref{theorem:convex_HC}.
Eventually, we found a self-contained proof and decided to split the content into two parts.}
suggests that \Cref{cor:emptykcycle} and a weakening of \Cref{cor:extend_to_2n_minus_3} might apply to the general case of simple drawings.
Hence, we conclude this article with the following two strengthenings of Rafla's~conjecture.

\begin{conjecture}
\label{conjecture:empty-k-cycles}
For every simple drawing $\calD$ of $K_{n}$ and every integer $3 \leq k \leq n$ there exists an empty $k$-cycle in~$\calD$.
\end{conjecture}

\begin{conjecture}
\label{conjecture:rafla_2n_plus_3}
Every simple drawing of $K_n$ with $n \geq 3$ contains a plane Hamiltonian subdrawing on $2n-3$ edges.
\end{conjecture}

\bibliography{references.bib}

\end{document}